\documentclass[12pt]{article}
\usepackage{amsmath}
\usepackage{graphicx}
\usepackage{enumerate}
\usepackage{natbib}
\usepackage{url} 

\usepackage{bm}
\usepackage{amssymb}

\newcommand{\indep}{\perp \!\!\! \perp}
\DeclareMathOperator*{\argmax}{arg\,max}

\usepackage{algorithm}
\usepackage[noend]{algpseudocode}
\usepackage{amssymb}
\usepackage{amsfonts}
\usepackage{mathrsfs}
\usepackage{amsthm}
\usepackage{multirow}
\usepackage{booktabs}
\usepackage{verbatim}
\usepackage{array}
\usepackage{caption}
\usepackage{setspace}
\usepackage{fixltx2e}
\usepackage[compact]{titlesec}
\usepackage{siunitx} 
\newcommand{\blind}{1}

\newtheorem{theorem}{Theorem}
\newtheorem{lemma}[theorem]{Lemma}

\begin{document}

\def\spacingset#1{\renewcommand{\baselinestretch}%
{#1}\small\normalsize} \spacingset{1}


\if1\blind
{
  \title{\bf Adaptive Weight Learning for Multiple Outcome Optimization with Continuous Treatment}
  \author{Chang Wang, Lu Wang\\
    Department of Biostatistics, University of Michigan}
  \maketitle
} \fi

\if0\blind
{
  \bigskip
  \bigskip
  \bigskip
  \begin{center}
    {\bf Optimal Dose Finding for Composite Outcomes/Optimization of Composite outcomes with Continuous Treatment}
\end{center}
  \medskip
} \fi

\bigskip

\begin{abstract}


To promote precision medicine, individualized treatment regimes (ITRs) are crucial for optimizing the expected clinical outcome based on patient-specific characteristics. However, existing ITR research has primarily focused on scenarios with categorical treatment options and a single outcome. In reality, clinicians often encounter scenarios with continuous treatment options and multiple, potentially competing outcomes, such as medicine efficacy and unavoidable toxicity. To balance these outcomes, a proper weight is necessary, which should be learned in a data-driven manner that considers both patient preference and clinician expertise. 
In this paper, we present a novel algorithm for developing individualized treatment regimes (ITRs) that incorporate continuous treatment options and multiple outcomes, utilizing observational data. Our approach assumes that clinicians are optimizing individualized patient utilities with sub-optimal treatment decisions that are at least better than random assignment. Treatment assignment is assumed to directly depend on the true underlying utility of the treatment rather than patient characteristics. The proposed method simultaneously estimates the weighting of composite outcomes and the decision-making process, allowing for construction of individualized treatment regimes with continuous doses. The proposed method's estimators can be used for inference and variable selection, facilitating the identification of informative treatment assignments and preference-associated variables. We evaluate the finite sample performance of our proposed method via simulation studies and apply it to a real data application of radiation oncology analysis.

\end{abstract}

\noindent%
\vfill

\newpage
\spacingset{1.9} 

\section{Introduction}

Precision medicine offers a more targeted and effective approach to healthcare by considering an individual's genetic makeup, lifestyle, and environmental factors \citep{precision_med}. Individualized Treatment Regimes (ITRs) play a crucial role in precision medicine, as they offer personalized recommendations for treatment based on a patient's unique characteristics. This helps to optimize the desired clinical outcome. However, a limitation of many ITRs is that they are designed to address a single outcome, whereas real-world scenarios often require balancing multiple, sometimes conflicting outcomes, such as medication efficacy and toxicity. Furthermore, they may not be suitable for continuous treatment scenarios. 

Despite extensive research on ITRs, much of the literature focuses on single-outcome scenarios. A variety of common estimators, such as Q-learning \citep{Qian_murphy,Laber_Q}, A-learning \citep{Alearning}, inverse propensity weighting \citep{IPW} and tree-based algorithm \citep{STRL,LZ,godotree}, can be used for optimizing treatment regimes with a single scalar objective function. When multiple outcomes are considered, a composite utility function is often used to simplify the problem into a single outcome scenario. Unfortunately, these composite utilities are typically unknown and may be hard to specify due to heterogeneity among patients and disagreement between patients and clinicians. To address this issue, a data-driven weight learning approach is desirable, which incorporates both patient preferences and clinician expertise to balance the weights of multiple outcomes.

Some existing methods for handling multiple outcomes, such as split questions, inverse reinforcement learning and bivariate continual reassessment method (bCRM), have limitations. For example, the question-splitting approach \citep{yingchao,laber_split} estimates individualized preferences through patient self-report questionnaires and then performs reinforcement learning with the estimated objective function. This approach neglects observed treatment assignment mechanisms, does not incorporate expert information, and requires a well-designed patient self-report questionnaire, which is often not feasible in real-world applications. Inverse reinforcement learning \citep{IRL_NG,IRL_Martin} also has limitations, as it assumes an unknown but fixed reward function and cannot incorporate patient-specific preferences in the model setting. Methods like bCRM \citep{bCRM} can balance medicine efficacy and toxicity, but they usually only focus on finding the maximum tolerated dose rather than optimal composite outcomes and patient preference is not considered. Besides, since our research interests are limited to observational study, discussions about clinical trial relevant methods like \cite{clinical} are omitted.

While Pareto optimality is commonly employed in multi-objective optimization, its utility in optimal dose finding is often limited. Pareto method maintains the separation of elements in the objectives throughout the decision optimization process \citep{gunantara} and leverage Pareto dominance to discern the Pareto front as potential solutions \citep{pareto}. Although these methods can reduce the number of decisions, they fall short in defining the singular optimal solution. Moreover, in clinical practice, multiple outcomes may exhibit monotonic behavior concerning the assigned dose. For instance, both toxicity and efficacy may increase with medication dose. This results in every dose being Pareto optimal, rendering it challenging to achieve a reduction in the decision space.

With the advancement of precision healthcare, ITRs with continuous dosage treatment are also becoming increasingly important. Examples of such scenarios include optimal dose finding in radiation oncology therapy and new drug trials. Although there are many algorithms designed for optimal dose finding \citep{CZK,LZ}, none of them can incorporate composite outcomes with an unknown utility function. This work is closely related to \cite{Laber_category}, but proposes a novel treatment assignment mechanism that allows for the optimization of composite outcomes with continuous treatment. The proposed algorithm also has a concise expression for estimation and inference.

The following pages are organized as follows, in Section 2, we propose a pseudo-likelihood method to estimate patient utility functions from observational data with categorical or continuous treatment and incorporate composite outcomes with unknown utility functions. We present theoretical results about the proposed method in Section 3, evaluate its performance in simulation experiments in Section 4, and apply it to a real-data application from the radiation oncology study in Section 5. Proofs and additional simulation results are provided in the appendix, along with a discussion of extending the method to more than two outcomes.

\section{Method}
Denote the independent and identically distributed observed data are $(\bm{X}_i,A_i,Y_i,Z_i), i=1,...,n$, where $\bm{X}_i \in \mathcal{X} \subseteq \mathbb{R}^p$ are patient covariates, $A_i \in \mathcal{A}$ is a continuous treatment or exposure, and $Y_i$ and $Z_i$ denote two real value outcome of interests for which higher values are more desirable. 

An individualized treatment rule is a function $d: \mathcal{X} \rightarrow \mathcal{A}$ such that patients with covariates $\bm{X} = \bm{x}$ will be assigned to treatment $d(\bm{x})$.

Let $Y^*(a)$ denotes the counterfactual outcome when patient takes the treatment $a$ and for any regime $d$, define $Y^*(d) = Y^*\{d(\bm{X})\}$. A regime for the outcome is define as the optimal $d_{Y}^{opt}$ when $\mathbb{E} Y^*(d_Y^{opt}) \geq \mathbb{E} Y^*(d)$ for any other regime $d$. The optimal regime $d_{Z}^{opt}$ for the outcome $Z$ is defined analogously. To identify optimal DTRs, we make the standard assumptions to link the distribution law of counterfactual data with that of observational data \citep{IPW}.

\newtheorem{assumption}{Assumption}

\begin{assumption}[Consistency]
The observed outcome is the same as the counterfactual outcome under the assigned treatment: $Y = Y^*(A),Z = Z^*(A),$.
\end{assumption}

\begin{assumption}[Positivity]
The conditional probability density function of treatment $f_{A\mid \bm{X}}(a\mid \bm{x}) \geq \epsilon >0$, $\forall (\bm{x},a) \in \mathcal{X} \times \mathcal{A}$. 
\end{assumption}

\begin{assumption}[Ignorability]
$A \indep (\mathcal{Y},\mathcal{Z}) \mid  \bm{X}$, where $\indep$ denotes statistical independence and $\mathcal{Y} = \{Y^*(a): a \in \mathcal{A} \}$, $\mathcal{Z} = \{Z^*(a): a \in \mathcal{A} \}$, also known as no unmeasured confounder assumption (NUCA).
\end{assumption}

We also assume the continuity of the counterfactual outcome mean, i.e., $\mathbb{E}Y^*(a)$ is continuous with respect to $a \in \mathcal{A}$. Furthermore, we require the Stable Unit Treatment Value Assumption (SUTVA), which means that there is no interference between units and there are no multiple versions of treatment \citep{SUTVA}.

Define $Q_Y(\bm{x},a) = \mathbb{E} (Y \mid \bm{X} = \bm{x}, A=a)$ as the conditional mean expectation. Then under the preceding assumptions, it can be shown that $d_Y^{opt}(\bm{x}) = \argmax_{a\in \mathcal{A}} Q_Y(\bm{x},a)$. Since $d_Y^{opt}(\bm{x})$ is not necessarily equal to $d_Z^{opt}(\bm{x})$, neither of them may be optimal when both outcomes are of interest. 

We assume there exists a composite outcome which consists of multiple outcomes and is possibly patient-specific. Denote the composite outcome as $U=u(Y,Z;X)$, where $u(\cdot,\cdot;\bm{X}): \mathbb{R}^2 \rightarrow  \mathbb{R}$ is a mapping from multiple outcomes to a scalar measurement of optimality and may be covariate-dependent on patient characteristics $\bm{X}$. The optimal treatment regime relevant to $U$ is thus defined as $d_U^{opt}(\bm{x})$, when $\mathbb{E}u(Y^*(d_U^{opt}),Z^*(d_U^{opt});\bm{x})\geq \mathbb{E}u(Y^*(d),Z^*(d);\bm{x})$, for any other regime $d$.

Here we consider the scenario that composite outcome is characterized as a weighted version of multiple outcome $Q_{\theta}(\bm{x},a) = w(\bm{x};\theta) Q_{Y}(\bm{x},a) + \{1-w(\bm{x};\theta)\} Q_{Z}(\bm{x},a)$, where the individualized weight $w(\bm{x};\theta)$ is parameterized by unknown parameter $\theta$.

To ensure that our proposed model is valid, we impose several key assumptions. First, we assume that clinicians optimize each patient's counterfactual outcomes after confirming their individualized preferences. We also assume that the observed clinicians' decisions are imperfect and at least better than random assignment. Second, we assume that the treatment assignment mechanism is a function of the composite outcome $Q_{\theta}(\bm{x},a)$, which is determined by both clinicians and patients, rather than an arbitrary function of patient characteristics. This assumption differs from the procedure proposed by \cite{Laber_category}, which assumes that $Pr\{A=d_U^{opt}(\bm{x})\mid \bm{X}=\bm{x}\}=expit(\bm{x}^T\beta)$ and can only be applied to categorical treatment.

We propose the following explicit expression for the observed treatment assignment made by clinicians:
\begin{equation}\label{A}
    f_{A\mid \bm{X}}(A=a \mid \bm{X} = \bm{x}) = 
\frac{exp\{ \beta  Q_{\theta}(\bm{x},a) \}}{ \int_{\mathcal{A}} exp\{\beta   Q_{\theta}(\bm{x},t) \} d\nu(t) },
\end{equation}
where $\beta$ is a scalar parameter measuring the optimality of observed assignment and $w(\bm{x};\theta)=expit(\bm{x}^T\theta) \in [0,1]$ is describing the individualized preference about multiple outcomes, which is an essential part of the composite outcome $Q_{\theta}(\bm{x},a)$. For continuous treatment, we define $\int_{\mathcal{A}} exp\{\beta   Q_{\theta}(\bm{x},t) \} d\nu(t) = \int_{\mathcal{A}} exp\{\beta   Q_{\theta}(\bm{x},t) \} dt$ and for categorical treatment, we define $\int_{\mathcal{A}} exp\{\beta   Q_{\theta}(\bm{x},t) \} d\nu(t) = \sum_{t \in \mathcal{A}} exp\{\beta   Q_{\theta}(\bm{x},t) \}$, so as to make the method compatible with both types of treatments.

Because the joint distribution is $f(\bm{X},A,Y,Z) = f(Y,Z \mid \bm{X},A) f(A \mid \bm{X}) f(\bm{X})$, assuming that $f(Y,Z \mid \bm{X},A) $ and $ f(\bm{X})$ is not relevant to the parameters of interest $(\theta,\beta)$, we propose the estimator $(\hat{\theta}_n, \hat{\beta}_n)$ of $(\theta, \beta)$ by maximizing the pseudo-likelihood 

$$ \mathcal{L}_n(\theta,\beta) = \prod_{i=1}^{n}  \frac{exp\{ \beta  Q_{\theta}(\bm{X}_i,A_i) \}}{ \int_{\mathcal{A}} exp\{\beta   Q_{\theta}(\bm{X}_i,t) \} d\nu(t) }.$$

Let $\hat{Q}_{Y}$ and $\hat{Q}_{Z}$ denote estimators of $Q_{Y}$ and $Q_{Z}$ obtained from model fitting to the observed data. For a fixed value of $\theta$, let $\hat{Q}_{\theta}(\bm{x},a) = w(\bm{x};\theta) \hat{Q}_{Y}(\bm{x},a) + \{1-w(\bm{x};\theta)\} \hat{Q}_{Z}(\bm{x},a)$ and correspondingly let $\hat{d}_{\theta}^{opt}(\bm{x}) = \argmax_{a\in \mathcal{A}} \hat{Q}_{\theta}(\bm{x},a)$ be the plug-in estimator of $d_{\theta}^{opt}$. The following estimated pseudo-likelihood is smooth about parameter $(\theta,\beta)$ and any standard gradient-based optimization algorithms can be used to solve the optimizer for $\hat{\mathcal{L}}_n(\theta,\beta)$.

$$ \hat{\mathcal{L}}_n(\theta,\beta) = \prod_{i=1}^{n}  \frac{exp\{ \beta  \hat{Q}_{\theta}(\bm{X}_i,A_i) \}}{ \int_{\mathcal{A}} exp\{\beta   \hat{Q}_{\theta}(\bm{X}_i,t) \} d\nu(t) }.$$

With the estimated $(\hat{\theta},\hat{\beta})$, we can estimate the individualized preference $w(\bm{x};\hat{\theta})$ and further the personalized composite outcome $\hat{Q}_{\hat{\theta}}(\bm{x},a) = w(\bm{x};\hat{\theta}) \hat{Q}_{Y}(\bm{x},a) + \{1-w(\bm{x};\hat{\theta})\} \hat{Q}_{Z}(\bm{x},a)$. Correspondingly, $\hat{d}_{\hat{\theta}}^{opt}(\bm{x}) = \argmax_{a\in \mathcal{A}} \hat{Q}_{\hat{\theta}}(\bm{x},a)$ is the finalized estimator for optimal treatment regime $d_{\theta}^{opt}$.

The algorithm to construct $(\hat{\theta},\hat{\beta})$ and $d_{\theta}^{opt}(\bm{x})$ is given in Algorithm 1 below.

\begin{algorithm} 
\caption{Individual treatment regime estimation for composite outcome}
\begin{algorithmic}[1] 

    \State  Fit conditional outcome mean models $\hat{Q}_Y(\bm{x},a)$ and $\hat{Q}_Z(\bm{x},a)$ respectively for 
    $$Q_Y(\bm{x},a) = \mathbb{E} (Y \mid \bm{X} = \bm{x}, A=a), Q_Z(\bm{x},a) = \mathbb{E} (Z \mid \bm{X} = \bm{x}, A=a).$$ 
    
    \State Use gradient-based optimization or other algorithms to solve the optimizer $(\hat{\theta},\hat{\beta})$ to maximize the pseudo-likelihood $\hat{\mathcal{L}}_n(\theta,\beta)$: 
    $$ \hat{\mathcal{L}}_n(\theta,\beta) = \prod_{i=1}^{n}  \frac{exp\{ \beta  \hat{Q}_{\theta}(\bm{X}_i,A_i) \}}{ \int_{\mathcal{A}} exp\{\beta   \hat{Q}_{\theta}(\bm{X}_i,t) \} d\nu(t) }.$$
    
    \State Estimate the composite outcome: $$\hat{Q}_{\hat{\theta}}(\bm{x},a) = w(\bm{x};\hat{\theta}) \hat{Q}_{Y}(\bm{x},a) + \{1-w(\bm{x};\hat{\theta})\} \hat{Q}_{Z}(\bm{x},a).$$
       
    \State Calculate the finalized estimator for optimal treatment regime $d_{\theta}^{opt}(\bm{x})$: 
    $$\hat{d}_{\hat{\theta}}^{opt}(\bm{x}) = \argmax_{a\in \mathcal{A}} \hat{Q}_{\hat{\theta}}(\bm{x},a).$$

\end{algorithmic}
\end{algorithm}

\begin{assumption} \label{general}

The following conditions hold.

1, $\beta \in \mathcal{B} \subset \mathbb{R}$ and $\theta \in \Theta \subset \mathbb{R}^p$, where $\mathcal{B}$ and $\Theta$ are compact.

2, The treatment relevant parameter $\beta > 0$ , which also means the treatment assignment is at least better than random assignment. Treatment that leads to a better composite outcome should have a higher probability to be chosen.   

3, The weight function $w(\bm{x}; \theta)$ is twice continuously differentiable with bounded first and second derivative $w'(\bm{x}; \theta)$ and $w''(\bm{x}; \theta)$.

4, Covariate vector $\bm{x} \in \mathcal{X} \subset \mathbb{R}^q$ and treatment $A \in \mathcal{A} \subset \mathbb{R}$, where $\mathcal{X}$ and $\mathcal{A}$ are compact. The outcomes $Q_Y(\bm{x},A)$ and $Q_Y(\bm{x},A)$ is bounded. 

\end{assumption}

\begin{assumption}[Consistent outcome estimator]\label{consistentQ}\  
$\sup_{a \in \mathcal{A}} | Q_{Y}(\bm{x},a) - \hat{Q}_{Y}(\bm{x},a)| \xrightarrow{p} 0$,
$\sup_{a \in \mathcal{A}} | Q_{Z}(\bm{x},a) - \hat{Q}_{Z}(\bm{x},a)| \xrightarrow{p} 0$ and
$\sup_{a \in \mathcal{A}} | Q_{\theta}(\bm{x},a) - \hat{Q}_{\theta}(\bm{x},a)| \xrightarrow{p} 0$, uniformly for $\bm{x} \in \mathcal{X}$. Additionally, we assume $\sup_{a \in \mathcal{A}} | Q_{\theta}(\bm{x},a) - \hat{Q}_{\theta}(\bm{x},a)| = o_p(1/\sqrt{n})$.
\end{assumption}

\begin{assumption}[Identifiable optimal treatment]\label{identifyQ}
 $\liminf_{n \rightarrow \infty} Q_{\theta}(\bm{x},a_n) \geq Q_{\theta}(\bm{x}, d_{\theta}^{opt}(\bm{x}))$ implies $\mid a_n - d_{\theta}^{opt}(\bm{x})\mid \xrightarrow{p} 0$ for any sequence $\{ a_n\} \in \mathcal{A}$.
\end{assumption}

\newtheorem*{remark}{Remark}

\begin{remark}
Assumption \ref{consistentQ} requires the convergence rate for the outcome fitting to be $o_p(n^{-\frac{1}{2}})$, which is typically satisfied when regular regression methods such as linear or generalized linear models are employed. While not essential for maximizing the ultimate value function, Assumption \ref{identifyQ} is necessary for identifying the optimal decision rule.
\end{remark}

\begin{assumption}[Identifiable patient preference]\label{identifyw}

Let $\mathcal{X}_S$ be the subsets of $\mathcal{X}$ defined as: $\mathcal{X}_S=\{\bm{x} \in \mathcal{X}: Cov_{A} 
[Q_{Y}(\bm{X},A),
 Q_{Z}(\bm{X},A)|\bm{X} ]$ is full rank \}, then $Pr(\bm{X} \in \mathcal{X}_S) > 0$ and $\mathbb{E} (\bm{X}\bm{X}^T|\bm{X} \in \mathcal{X}_S)$ is full rank.
\end{assumption}

\begin{remark}
    For patients who do not belong to $\mathcal{X}_S$, it can be demonstrated that their two outcomes are either perfectly correlated, or one outcome is constant with respect to treatment. Consequently, the problem of multiple outcomes can be reduced to a single outcome, and no preference information can be obtained.
\end{remark}

\begin{theorem}[\textbf{Identifiability}] \label{identify}

Under Assumption \ref{identifyw}, $(\theta_0,\beta_0)$ is uniquely identifiable under the model given by $\mathcal{L}_n(\theta,\beta)$.
\end{theorem}

\section{Theorem}

\begin{assumption}\label{continuous_log}
$\sup_{(\theta,\beta)} \mid (\mathbb{E}_n - \mathbb{E}) \hat{m}(\bm{X},A; \theta, \beta) \mid \xrightarrow{p} 0$, where $\hat{m}(\theta, \beta) := \hat{m}(\bm{X},A; \theta, \beta) = \beta  \hat{Q}_{\theta}(\bm{X},A) - 
    log\int_{\mathcal{A}} exp\{\beta   \hat{Q}_{\theta}(\bm{X},t) \} d\nu(t)$ is the observed log likelihood.
\end{assumption}

\begin{remark}
    It can be shown that when the treatment only consists of a finite set of dose options, Assumption \ref{continuous_log} is satisfied. Although when the treatment dose can take on any continuous value, Assumption \ref{continuous_log} is no longer trivial. Nevertheless, in practical applications, a finite grid is often set for continuous doses, rendering this assumption still reasonable.
\end{remark}

Along with the above assumptions, we assume Equation \ref{A} is the ground truth treatment assignment mechanism and we correctly specify the patient preference model $w(\bm{x};\hat{\theta})$. The following result states the consistency of the maximum pseudo-likelihood estimators for the utility function parameter $\theta$ and treatment assignment parameter $\beta$.

\begin{theorem}[\textbf{Consistency}] \label{consistency}
Let the maximum pseudo likelihood estimators be as in section xxx, $(\hat{\theta}_n, \hat{\beta}_n) = \argmax_{\theta \in \mathbb{R}^p, \beta \in \mathcal{B}} \hat{\mathcal{L}}_n(\theta,\beta)$. Assume that $\mathcal{B}$ is a compact set with $\beta_0 \in \mathcal{B}$ and that $\lVert \mathbb{E} \bm{X} \rVert < \infty$. Then $\lVert \hat{\theta}_n -\theta_0 \rVert \xrightarrow{p} 0$ and $\lVert \hat{\beta}_n - \beta_0\rVert \xrightarrow{p} 0$ as $n \rightarrow \infty$, where $\lVert .\rVert$ is the Euclidean norm.
\end{theorem}

Denote the value function $V_{\theta_0}(d) = \mathbb{E} \big[ Q_{\theta_0}\{ \bm{X},A=d(\bm{X})\} \big]$ as the expectation of composite outcome when all the patients take treatment policy $d(\bm{x})$. The following result shows the consistency of the value function under the estimated optimal policy.

\begin{theorem}[\textbf{Value consistency with patient-specific utility}]\label{value}
    Let $\hat{\theta}_n$ be the maximum
pseudo-likelihood estimator for $\theta$ and let $\hat{d}_{\hat{\theta}_n}$ be the associated estimated optimal policy. Then, under the given assumptions, $|  V_{\theta_0}(\hat{d}_{\hat{\theta}_n}) - V_{\theta_0}(d_0^{opt})| \overset{p}{\to} 0$ as $n \to \infty$. Along with Assumption \ref{identifyQ}, $|\hat{d}_{\hat{\theta}_n}(\bm{x}) -d_0^{opt}(\bm{x}) | \overset{p}{\to} 0, \forall \bm{x} \in \mathcal{X}$.
\end{theorem}

In addition to Assumption \ref{identifyw}, we further use the following assumption to make parameter inference applicable.

\begin{assumption}\label{inverse}
   $\mathbb{E}_{\bm{X}}  Cov_{A} 
[Q_{\theta_0}(\bm{X},A),
 \beta_0R(\bm{X},A)w'(\bm{X};\theta_0)
|\bm{X} ]$ is full rank. 
\end{assumption}

\begin{theorem}[\textbf{Asymptotic Normality}] \label{normal}

With the given regularity assumptions,
$\hat{\beta}_n$ and $\hat{\theta}_n$ is asymptotically linear and normal, i.e., $\sqrt{n} (\hat{\beta}_n - \beta_0,
 \hat{\theta}_n - \theta_0)$ can be expressed as:
$$  B^{-1}
\Bigg( 
    \begin{array}{c}
       n^{-\frac{1}{2}}  \sum_{i=1}^n \left\{ Q_{\theta_0}(\bm{X}_i,A_i) -\mathbb{E}\left[  Q_{\theta_0}(\bm{X}_i,A)   \mid \bm{X}_i \right] \right\} \\
n^{-\frac{1}{2}} \sum_{i=1}^n \beta_0\left\{ R_iw'_i(\theta_0) -\mathbb{E}\left[ R_i w'_i(\theta_0)  \mid \bm{X}_i \right] \right\}
    \end{array}
\Bigg) + o_p(1),
$$
where $B = \mathbb{E}_{\bm{X}}  Cov_{A} 
[Q_{\theta_0}(\bm{X},A),
 \beta_0R(\bm{X},A)w'(\bm{X};\theta_0)
|\bm{X} ]$, $R_i = R(\bm{X}_i,A_i) = Q_{Y}(\bm{X}_i,A_i) - Q_{Z}(\bm{X}_i,A_i)$ and $w'_i(\theta_0) = w'(\bm{X}_i;\theta_0)$.

The asymptotic distribution of $(\hat{\beta}_n ,
 \hat{\theta}_n )$ is
$\sqrt{n}( \hat{\beta}_n - \beta_0, \hat{\theta}_n - \theta_0)
\sim N(0, B^{-1}) $.
\end{theorem}

\begin{remark}

The asymptotic variance expression shows that a random decision ($\beta_0=0$) would violate Assumption \ref{inverse} and render both estimation and inference inapplicable. This is due to the fact that a random decision would provide no information about patient preference. Conversely, a larger value of $\beta_0$ would reduce the asymptotic variance as it would include more relevant information in the observed treatment assignment, leading to better estimation and inference.

\end{remark}

\section{Simulation Experiments}

Let $\bm{X}=(X_1,X_2)$ be a vector of independent and identical random variables following normal distribution with mean 0 and standard deviation 0.5. Since we are interested in the algorithm performance when the treatment is continuous, we let the multiple outcomes be $Y = A(4X_1-2X_2+2)-2A^2 + \epsilon_Y$ and $Z = A(2X_1-4X_2-2)-2A^2 + \epsilon_Z$, where $\epsilon_Y$ and $\epsilon_Z$ are independent normal noises with mean 0 and standard deviation 0.5. Then the optimal treatment regime is $d_Y^{opt}(\bm{X}) = (2X_1-X_2+1)/2$ and $d_Z^{opt}(\bm{X}) = (X_1-2X_2-1)/2$.

For fixed utility scenario, we assume patient preference is constant across the population and let the composite outcome be a fixed function of multiple outcomes, i.e., $Q_{\theta}(\bm{X},A) = \omega Q_{Y}(\bm{X},A) + (1-\omega)Q_{Z}(\bm{X},A)$. For patient-specific utility scenario, we assume the composite outcome is $Q_{\theta}(\bm{X},A) = w(\bm{X};\theta)Q_{Y}(\bm{X},A) + \{1-w(\bm{X};\theta)\}Q_{Z}(\bm{X},A)$, where $w(\bm{X};\theta) = expit(\theta_0 + \theta_1X_1+\theta_2X_2)\in (0,1)$.

We estimate $Q_Y$ and $Q_Z$ with linear models, assuming the model is correctly specified. Multiple scenarios with difference setting of parameters are implemented and 500 Monte Carlo replications are performed per scenario for all the simulations.

\begin{table}[htb]
\centering
\begin{tabular}{cc|cc|cc|c}
\hline
\multicolumn{2}{c|}{Parameter} & \multicolumn{2}{c|}{Estimation}& \multicolumn{2}{c|}{Type I Error} &\multicolumn{1}{c}{Power} \\
n&  $(\beta_0,\omega_0)$ & $\hat{\beta}$& $\hat{\omega}$ & $\beta=\beta_0$&  $\omega=\omega_0$ & $\beta=0$\\
\hline
100  & (0.15, 0.3)& 0.159 (0.085) &0.322 (0.275) &0.07    &0.016  &0.498  \\
     & (0.25, 0.3)& 0.259 (0.093) &0.286 (0.189) &0.044   &0.028  &0.880   \\
\hline
250  & (0.15, 0.3)&0.152 (0.054) &0.303 (0.200) &0.062   &0.034  &0.882 \\
     & (0.25, 0.3)&0.251 (0.061) &0.294 (0.130) &0.048   &0.046  &0.998    \\
\hline
500  & (0.15, 0.3)&0.152 (0.037) &0.300 (0.143)  &0.056    &0.054    &0.996         \\
     & (0.25,0.3) &0.252 (0.042) &0.300 (0.089)   &0.056    & 0.062   &1.000\\
\hline
\end{tabular}
\caption{Estimation and inference results for simulations with fixed utility scenario. }
\label{total_cell}
\end{table}

Table 1 presents the performance of estimation and inference in fixed utility scenarios. The mean estimates of $\beta_0$ and $\omega_0$, along with their standard deviations across replications, type one error, and power against the alternative hypothesis $\beta_0 = 0$ are reported. The simulation results indicate that the estimates of the two parameters are consistent, with type I error rates close to the nominal level of 0.05 and acceptable power even in small sample sizes. As the sample size increases, the standard errors decrease, and the type I error becomes more centered around 0.05, as predicted by the central limit theorem. Additionally, as $\beta_0$ increases, suggesting an increase in optimal observed treatment assignment, the standard error of $\hat{\omega}$ decreases approximately at a rate of $1/\beta_0$, which is in line with the presented theorem and indicates a better estimation of multiple outcome weights.

\begin{table}[htb]
\centering
\begin{tabular}{cc|ccccc}
\hline
n& $(\beta_0,\omega_0)$  & Optimal & New & Y Optimizer &  Z Optimizer & Observed\\
\hline
100  & (0.15, 0.3)    &0.64 (0.08) &0.48 (0.24) &-0.33 (0.16) &0.47 (0.10) &-0.39 (0.13)\\
     & (0.25, 0.3)    &0.64 (0.08) &0.56 (0.14) &-0.33 (0.16) &0.47 (0.10) &-0.25 (0.13)\\
\hline
250  & (0.15, 0.3)   &0.64 (0.05) &0.56 (0.12) &-0.33 (0.09) &0.47 (0.06) &-0.40 (0.09)\\
     & (0.25, 0.3)   &0.64 (0.05) &0.60 (0.07) &-0.33 (0.09) &0.47 (0.06) &-0.25 (0.08) \\
\hline
500  & (0.15, 0.3) &0.64 (0.04) &0.60 (0.07) &-0.32 (0.06) &0.47 (0.04) &-0.39 (0.06)\\
     & (0.25, 0.3) &0.64 (0.04) &0.62 (0.05) &-0.32 (0.07) &0.47 (0.04) &-0.25 (0.06)\\
\hline
\end{tabular}
\caption{Composite outcome value results for simulations with fixed utility scenario. }
\label{total_cell}
\end{table}

Table 2 displays the composite outcome values under different policies. It compares the true optimal policy, the newly proposed policy with estimated utility functions, policies estimated to maximize the two outcomes individually (represented by fixing $\omega$ = 1 and $\omega$ = 0), and the observed standard of care, with the standard deviation across replications shown in parentheses. The value of the standard care is the mean composite outcome under the generative model. Table 2 shows that as the sample size or $\beta_0$ increases, the composite outcome values under the proposed policy approach the optimal upper bound, with decreasing standard deviations. The new algorithm consistently outperforms the single outcome optimization strategy, even with small sample sizes. These results demonstrate the proposed algorithm's ability to effectively optimize the composite outcome in real-world scenarios. In the scenarios presented in Table 2, the composite outcome value under standard care increases as the sample size increases, but still falls behind the policy maximizing the Z outcome only. This result is due to a high weight assigned to the Z outcome ($1-\omega_0 = 0.7$). Despite this, the proposed policy continues to outperform both the standard care and individual outcome optimization strategies.

\begin{table}[htb]
\centering
\begin{tabular}{ccc|cccc}
\hline
$\bm{\theta}=(0,-2,2)$ &  Size & $\beta_0$ & $\hat{\beta}$& $\hat{\theta}_0$ & $\hat{\theta}_1$&  $\hat{\theta}_2$ \\
\hline
Estimates   & 500  & 0.40  & 0.40 (0.07) &0.02* (0.51) &-2.20* (1.06) & 2.13* (1.05)  \\
            &      & 0.60  & 0.61 (0.08) &0.02 (0.24) &-2.09 (0.60) &2.05 (0.63) \\      
            & 750  & 0.40  & 0.41 (0.06) &-0.007 (0.28) &-2.13 (0.82) &2.14 (0.76)  \\     
            &      & 0.60  & 0.61 (0.06) &-0.004 (0.19) &-2.05 (0.50) &2.07 (0.47)  \\      
            & 1000 & 0.40  & 0.40 (0.05) &0.01 (0.24) &-2.08 (0.61) &2.03 (0.60)  \\        
            &      & 0.60  & 0.61 (0.06) &0.01 (0.16) &-2.05 (0.42) &2.01 (0.40)   \\    
\hline
Type I Error   & 500  &  0.40 & 0.062 &0.024 &0.032 &0.032\\
               &      &  0.60 & 0.052 &0.040 &0.036 &0.046 \\  
               & 750  &  0.40 & 0.042 &0.020 &0.036 &0.042 \\                                 
               &      &  0.60 & 0.038 &0.030 &0.048 &0.046 \\                                 
               & 1000 &  0.40 & 0.082 &0.036 &0.040 &0.040  \\                                  
               &      &  0.60 & 0.086 &0.044 &0.044 &0.042\\                                \hline
\end{tabular}
\caption{Estimation and inference results for simulations with patient-specific utility scenario. }
\label{total_cell}
\end{table}


In Table 3, we present the results of a simulation study that evaluates the performance of estimation and inference in patient-specific utility scenarios. The weight function $\omega(\bm{X};\theta_0)$ is set to be negatively associated with patient characteristic $X_1$ and positively associated with $X_2$. Table 3 reports the mean estimates of $\beta_0$ and $\theta_0$, along with their standard deviations across replications, type one error, and power against the alternative hypothesis that the parameters equal zero. The simulation results indicate that the parameter estimates are consistent, with type I error rates close to the nominal level of 0.05 and acceptable power. Similar to the fixed utility scenario, the model becomes more efficient as the sample size or $\beta_0$ increases. It is important to note that more samples are required in patient-specific utility scenarios to distinguish heterogeneous patient preferences compared to weight estimation in fixed utility scenarios. In the scenarios presented in Table 3, when the sample size is 500, the model may not converge due to the singularity of the estimated $B$ matrix, which violates Assumption \ref{identifyw} and \ref{inverse} for observed data. In three out of 500 replications, the parameter estimates are extreme, but the composite outcome values are still reasonable. These three replications are marked with an asterisk in Table 3 and removed, but kept in Table 4.

\begin{table}[htb]
\centering
\begin{tabular}{cc|ccccc}
\hline
n& $\beta_0$  & Optimal & New & Y Optimizer &  Z Optimizer & Observed\\
\hline
500  & 0.40& 1.72 (0.16) &1.47 (0.44) &-3.61 (0.31) &-3.59 (0.33) &-4.14 (0.37)\\
     & 0.60& 1.72 (0.16) &1.62 (0.19) &-3.61 (0.31) &-3.59 (0.33) &-3.28 (0.35)\\
750  & 0.40& 1.71 (0.13) &1.57 (0.18) &-3.62 (0.26) &-3.60 (0.27) &-4.13 (0.30) \\ 
     & 0.60& 1.71 (0.13) &1.64 (0.14) &-3.62 (0.26) &-3.60 (0.27) &-3.28 (0.28) \\
1000  & 0.40& 1.70 (0.11) &1.60 (0.15) &-3.62 (0.22) &-3.62 (0.22) &-4.14 (0.27)\\ 
     & 0.60& 1.70 (0.11) &1.65 (0.12) &-3.62 (0.22) &-3.62 (0.22) &-3.28 (0.25)\\
\hline
\end{tabular}
\caption{Composite outcome value results for simulations with patient-specific utility scenario. }
\label{total_cell}
\end{table}

Table 4 presents the composite outcome values for patient-specific utility scenarios under different policies. Similar to the fixed utility scenario, we observe that the composite outcome values under the proposed policy converge towards the optimal upper bound as the sample size or $\beta_0$ increases, while the standard deviations decrease. We also observe that the policy of maximizing a single outcome is becoming less favorable due to the heterogeneity of patient preferences. 

These results provide valuable insights into the evaluation of model performance under different utility scenarios, with a focus on parameter estimation and composite outcome values.

\section{Application to Optimal Radiation Dose Finding}

We applied the proposed methods to the radiation oncology data for optimal dose finding. Two outcomes are taken into consideration when optimizing patients' health: patient toxicity and lesion local progression. Patients are defined to have toxic event when 6-month post treatment ALBI score increase more than 0.5 unit, when compared with baseline ALBI score. When 6-month results are missing, 3-month ALBI score are used. Local progression is defined as the recurrence contiguous with the resection cavity or original tumor site. Because it is a lesion level outcome while we are interested in patient level preference, we used ``any local progression'' as the second outcome. A common tradeoff raises from the fact that when radiation dose increase, local progression rate will decrease while toxicity rate will increase. 

In real clinical application, a utility table can be used to balance two outcomes for optimal dose finding. The best result without either local progression or toxicity is set to be 100\% and the worse result with both local progression and toxicity is set to be 0\%. The utility score for the other two scenarios are usually pre-specified based on expert knowledge, clinician suggestion or patient-specific preference. In our dataset, the following utility table was regarded as the ground truth when assigning optimal treatment, which is also equivalent to the definition $Utility = 1 - \text{Toxicity} \times \omega - \text{Local Progression} \times (1-\omega)$, where the true value for weight parameter is $\omega_0 = 0.6$.

\begin{center}
\begin{tabular}{ |c|c|c|c| } 
\hline
 Utility & Local Control & Local Progression \\
\hline
No Toxicity& 100\% & 60\%  \\ 
Toxicity & 40\% & 0\% \\ 
\hline
\end{tabular}
\end{center}

Two dose response curves are fitted for two outcomes separately. Mean liver dose (MLD) and biological effective dose (BED) are used as the main predictors for patient level toxicity and lesion level local progression respectively. Each patient received five or six fractions in this trial and dose per fraction and total dose are used to calculate the MLD/BED value. ``Ratio'' is a patient specific parameter to bridge the total dose and liver dose because liver and lesion receive different radiation doses under the clinical treatment and one patient can have multiple lesions at different locations. Logistic regression and survival analysis with weibull model were conducted to predict the probability of patient toxicity and one-year local progression for any lesions.

\begin{itemize}
\item MLD = Total Dose $\times$ Ratio $\times$ (Dose per Fraction $\times$ Ratio + 2.5)/ (2 + 2.5).
\item BED = Total Dose $\times$ (Dose per Fraction/10 + 1).
\item Total Dose  = Number of Fraction $\times$ Dose per Fraction.
\end{itemize}

We analyzed the data using the proposed method for optimizing multiple outcomes. We estimated decision rules where the weight/preference is a fixed scalar and where the weight is patient specific. Possible factors include worse biomarker condition (like child pugh score, albi Score), larger tumor number/size and other medical problems may influence patient preferences and vote for less aggressive treatment (less efficacy less toxicity). We conducted statistical inference to explore whether these variables are significant indicator for different outcome preference/weight.

\subsection*{Fixed Utility Analysis}
When the outcome weight is assumed to be a constant, i.e., $\text{Utility} = 1 - \text{Toxicity} \times \omega - \text{Local Progression} \times (1-\omega)$, the estimated weight is $\hat{\omega} = 0.55$, with 95\% confidence interval [0.525, 0.586], which is pretty close to the ground truth value ($\omega_0 = 0.60$) recommended by clinical experts. The estimation for parameter of assignment optimality is $\hat{\beta} = 28.8$, with 95\% confidence interval [20.5, 37.0], which is significantly higher than 0 and indicates that the observed treatment is significantly better than random decision. The performance of newly proposed algorithm is compared with observed treatment and random decision is used as benchmark for utility score comparison. It can be seen from Table \ref{fig:table_fix} that the new algorithm achieved a remarkable improvement on the utility score, toxicity rate and local progression rate.

\begin{table}[h!]
  \begin{center}
    \caption{\label{fig:table_fix}Model comparison for fixed utility analysis: The rate of toxicity and local progression are compared for three decision rules. Utility score are calculated with ground truth $\omega_0 = 0.6$. The percentage improvement measures the gained information by our algorithm and is defined as $(V_{new} - V_{observed}) / (V_{observed} - V_{random})$, where V denote the outcome to be compared. }
    \label{tab:table1}
    \begin{tabular}{ccccc}
      \toprule 
      \textbf{Outcomes} & \textbf{New} & \textbf{Observed}
      & \textbf{Random} & \textbf{\% Improvement} \\
      \midrule 
      Toxicity & 15.2\% & 15.4\% &17.7\%  & 9.0\%\\
      Local Progression   & 17.3\% & 20.9\% & 35.6\% & 24.9\% \\
      Utility & 84.0\% & 82.4\% &75.1\%  & 21.9\% \\
      \bottomrule 
    \end{tabular}
  \end{center}
  
\end{table}

\subsection*{Patient Specific Utility Analysis}

We also conducted analysis where the outcome weight is assumed to be patient specific, i.e., $Utility = 1 - \text{Toxicity} \times w(X) - \text{Local Progression} \times [1-w(X)]$. Four variables that might influence outcome weighting were considered: lesion number, baseline ALBI score, maximal diameter of lesions and total lesion volume. For ease of interpretation, continuous variables were dichotomized with population mean as the threshold. Table \ref{fig:table_pt_spe} shows that patients with larger tumor, in the sense of diameter (p = 0.07) or volume (p=0.10), is associated with a lower weight for toxicity and a higher weight for local progression outcome. The result is significant with a significance level $\alpha = 0.1$ for our small sample analysis (80 patient and 121 lesions) and is consistent with the common sense: more weight should be given to tumor control/intervention efficacy for patients with worse health condition, rather than toxicity or other chronic diseases.

\begin{table}[h!]
  \begin{center}
    \caption{\label{fig:table_pt_spe}Patient specific utility analysis: outcome weight, 95\% confidence interval and p value are calculated for each covariate and subgroup to study whether the outcome weight depends on patient characteristics.}
    \label{tab:table1}
    \begin{tabular}{ccccc}
      \toprule 
      \textbf{Covariate $X$} & \textbf{Group} & \textbf{$\hat{w}(x)$}
      & \textbf{95\% CI} & \textbf{p value} \\
      \midrule 
      \multirow{ 2}{*}{No. Lesion} & $= 1$    & 0.52  & [ 0.46, 0.58 ]  &
      \multirow{ 2}{*}{0.50}  \\
       & $\geq 2$ & 0.57  & [ 0.52, 0.62 ]  &  \\
      \hline
      \multirow{ 2}{*}{ALBI Score} & $< -2$    & 0.55  & [ 0.50, 0.59 ]  &
      \multirow{ 2}{*}{0.90}  \\
       & $\geq -2$ & 0.56  & [ 0.49, 0.63 ]  &  \\
      \hline
      \multirow{ 2}{*}{Lesion Size}  & $<5cm$ & 0.65  & [ 0.58, 0.71 ]  & \multirow{ 2}{*}{0.07}   \\
       & $\geq 5cm$ & 0.51  & [ 0.47, 0.54 ]  &  \\
       \hline
      \multirow{ 2}{*}{Lesion Volume}  & $<75cm^3$ & 0.64  & [ 0.57, 0.70 ]  & \multirow{ 2}{*}{0.10}   \\
       & $\geq 75cm^3$ & 0.51  & [ 0.47, 0.55 ]  &  \\
      \bottomrule 
    \end{tabular}
  \end{center}
\end{table}

\begin{figure}[H]
\caption{Different outcome preference among two groups: patients with larger tumor will assign more weight on local progression. Curves is visualizing the fact that as radiation dose, toxicity rate increase and local progression rate decrease, and black dots are the observed dosage assignment. Dash lines are the utility contour line, whose slopes are decided by the outcome weight.}%
    \label{fig:example}%
    \centering{{\includegraphics[width=7.5cm]{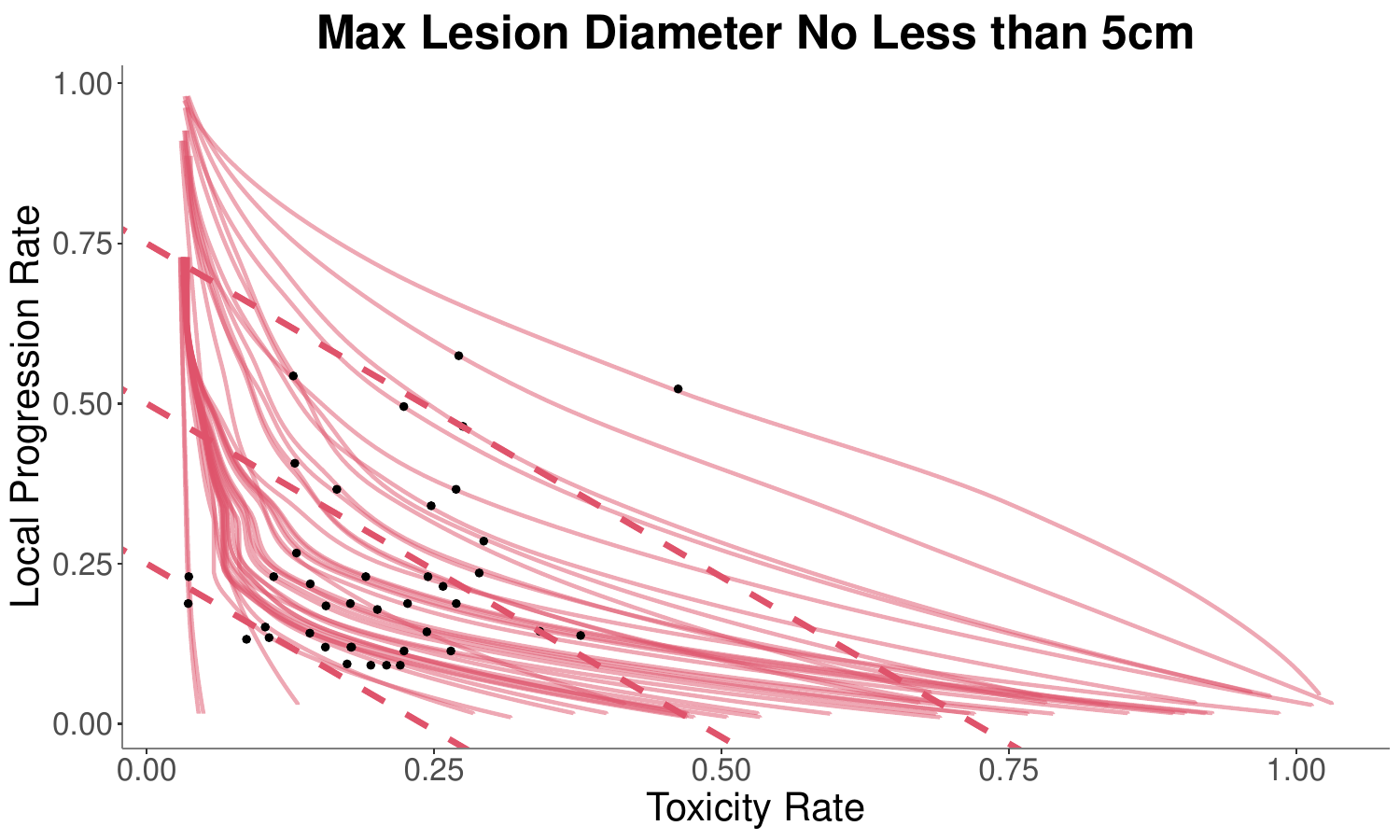} }}%
    \qquad{{\includegraphics[width=7.5cm]{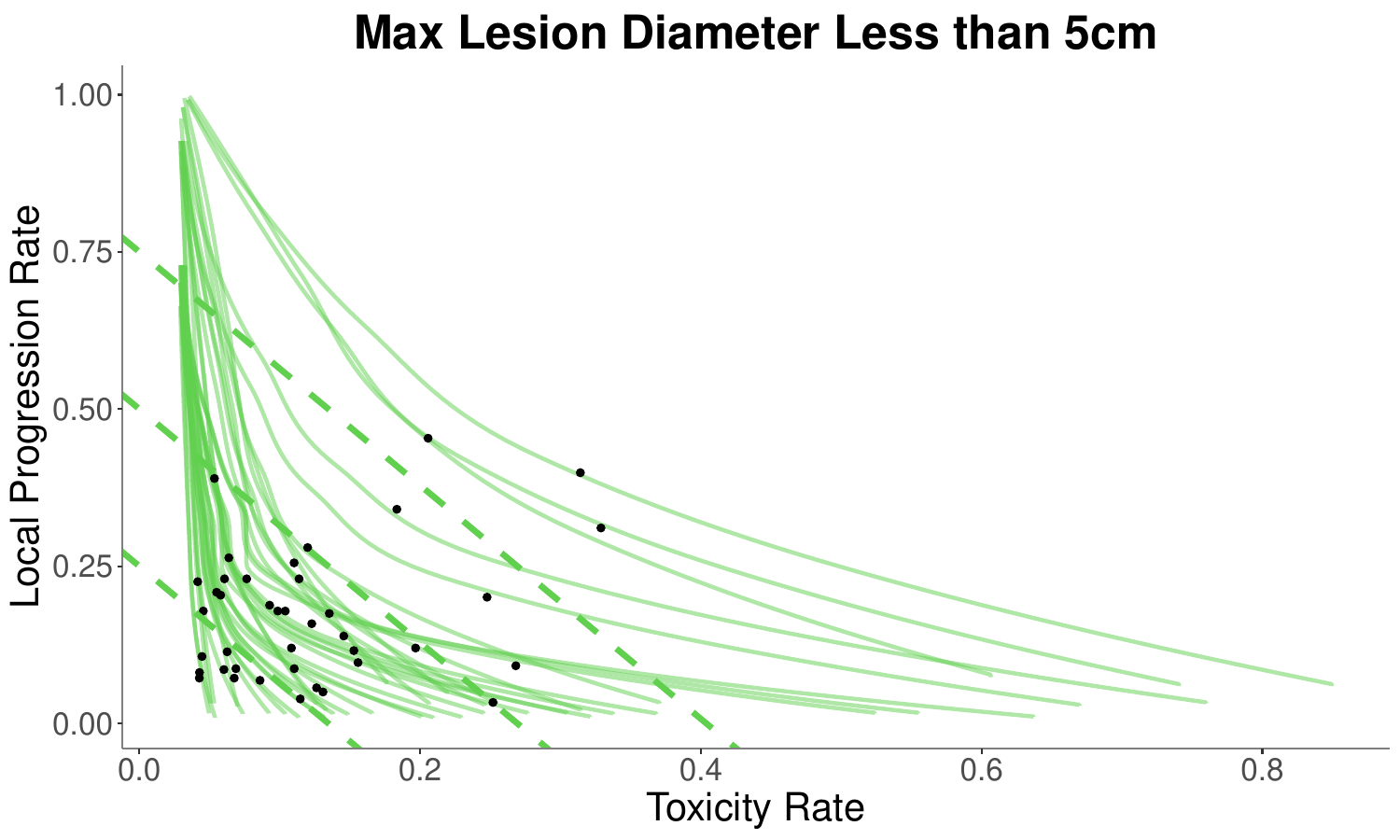} }}%
    
\end{figure}

\section{Discussion}

The use of individualized treatment rules with fixed scalar outcomes has been extensively studied, but in real-world clinical practice, clinicians must consider multiple secondary outcomes such as toxicity, mortality, and chronic pain, in addition to the primary outcome. This creates a gap between available statistical methods and clinical practice. To address this gap, \cite{Laber_category} proposed a method for handling scenarios with multiple outcomes and binary/categorical treatments that considers treatment assignment as a function of patient characteristics, i.e., $Pr\{ A = d^{opt}_U(\bm{x})|\bm{X} = \bm{x}\} = expit(\bm{x}^T\beta)$. However, this method fails to utilize information on counterfactual outcomes. 

In response to this limitation, we propose a new algorithm for incorporating scenarios with multiple outcomes and continuous treatment, assuming treatment assignment is based on counterfactual outcomes, i.e., $P(A=a|\bm{X} = \bm{x}) \propto exp\{ \beta Q_{\theta}(\bm{x},a) \}$. Our approach reduces the need for model specification and implicitly assumes that patient characteristics are independent of treatment assignment, conditional on the same composite outcome. The intuition behind this assumption is that if two doses result in the same composite outcome, they should have an equal chance of being selected after negotiation between clinicians and patients. This property of the treatment assignment mechanism should be fully utilized.

Another fundamental assumption is that clinicians and patients are making sub-optimal decisions, at least ones that are better than random choice. This assumption is necessary, as otherwise, no useful information can be extracted from random assignment. To distinguish sub-optimal decisions from random decisions, we make inference about the parameter $\beta$. A significantly positive $\beta$ can reject the random decision hypothesis and allow for further analysis of patient preference.

The main objective of this paper is to bridge the gap between available ITR methods and multiple outcomes applications. Therefore, we assume that the conditional outcome mean model is always known or correctly specified.  A possible future research direction could be to investigate the combination with missing data issues and competing risks in survival analysis.


\section{Appendix}
\renewcommand{\proofname}{\textbf{Proof}}

\begin{proof} [\textbf{Proof of Theorem \ref{identify}} ]

Recall the conditional distribution of treatment assignment is
$$f(A=a \mid \bm{X} = \bm{x};\beta,\theta) = 
\frac{exp\{ \beta  Q_{\theta}(\bm{x},a) \}}{ \int_{\mathcal{A}} exp\{\beta   Q_{\theta}(\bm{x},t) \} d\nu(t) }$$
Suppose $f(A \mid \bm{X};\beta,\theta) = f(A \mid \bm{X};\beta_0,\theta_0)$, $\forall (\bm{X},A)\  a.s.$, we have that
\begin{align*}
    &f(A \mid \bm{X};\beta,\theta) = f(A \mid \bm{X};\beta_0,\theta_0), \ \forall (\bm{X},A)\  a.s.\\
    \Longleftrightarrow \ & 
    \beta Q_{\theta}(\bm{X},a) - \beta Q_{\theta}(\bm{X},b) = 
    \beta_0 Q_{\theta_0}(\bm{X},a) - \beta_0 Q_{\theta_0}(\bm{X},b),\ \forall a,b \in \mathcal{A}, \ \forall \bm{X}\  a.s.\\
    \Longleftrightarrow \ & 
    \big[Q_Y(\bm{X},A) - Q_Z(\bm{X},A) \big]   \big[ \beta_0 w(\bm{X};\theta_0) - \beta w(\bm{X};\theta) \big] 
    + Q_Z(\bm{X},A)  (\beta_0 - \beta) = c(\bm{X})\\
    \tag*{\llap{$\forall (\bm{X},A)\  a.s.$}}
\end{align*}

For $\forall \bm{X} \in \mathcal{X}_S$, if there exists $(\lambda_1,\lambda_2)\not = (0,0)$ such that $\lambda_1 \big[Q_Y(\bm{X},A) - Q_Z(\bm{X},A) \big] + \lambda_2 Q_Z(\bm{X},A)$ equals to a constant irrelevant to $A$, then $Cov_{A} [Q_{Y}(\bm{X},A), Q_{Z}(\bm{X},A)|\bm{X} ]$ is not full rank, which contradicts Assumption \ref{identifyw}. Thus such $(\lambda_1,\lambda_2)$ must be $(0,0)$, which corresponds to $\beta = \beta_0 \not = 0$ (Assumption \ref{general}) and $w(\bm{X};\theta) = w(\bm{X};\theta_0)$, $\forall \bm{X} \in \mathcal{X}_S$.

Additionally, according to Assumption \ref{identifyw}, $\mathbb{E} \left[  \bm{X}\bm{X}^T  | \bm{X} \in \mathcal{X}_S\right]$ is a full rank matrix, and we can straightforwardly show that ``$w(\bm{X};\theta) = w(\bm{X};\theta_0)$, $\forall \bm{X} \in \mathcal{X}_S$'' is equivalent to ``$\theta = \theta_0$'', when the patient preference model is $w(\bm{X};\theta) = 1/[1+exp(\bm{X}^T\theta)]$.
\begin{align*}
    &w(\bm{X};\theta) = w(\bm{X};\theta_0), \ \forall \bm{X} \in \mathcal{X}_S\\
    \Leftrightarrow \ & \bm{X}^T\theta = \bm{X}^T\theta_0, \ \forall \bm{X} \in \mathcal{X}_S\\
    \Leftrightarrow \ & \mathbb{E} \left[ (\theta - \theta_0)^T \bm{X}\bm{X}^T (\theta - \theta_0) | \bm{X} \in \mathcal{X}_S\right]  = 0\\
    \Leftrightarrow \ & (\theta - \theta_0)^T \mathbb{E} \left[  \bm{X}\bm{X}^T  | \bm{X} \in \mathcal{X}_S\right] (\theta - \theta_0)  = 0\\
    \Leftrightarrow \ & \theta = \theta_0.
\end{align*}

Above all, $f(A \mid \bm{X};\beta,\theta) = f(A \mid \bm{X};\beta_0,\theta_0), \ \forall (\bm{X},A)\  a.s.$, implies $(\beta,\theta)=(\beta_0,\theta_0)$.

\end{proof}

\begin{lemma}[Kosorok 2008 Theorem 2.12] \label{Kosorok2.12}
Assume that (i) $\liminf_{n \rightarrow \infty} M(\theta_n) \geq M(\theta_0)$ implies $d(\theta_n,\theta_0)\xrightarrow{p} 0$ for some $\theta_0 \in \Theta$ and any sequence $\{ \theta_n\} \in \Theta$ (identifiability condition), where $d$ is a distance measure. Then for a sequence of estimator $\hat{\theta}_n \in \Theta$, if  (ii) $M_n(\hat{\theta}_n) = \sup_{\theta \in \Theta} M_n(\theta) - o_p(1)$ and (iii) $\sup_{\theta \in \Theta} \mid M_n(\theta) - M(\theta) \mid \xrightarrow{p} 0$ then $d(\hat{\theta}_n,\theta_0) \xrightarrow{p} 0$.
\end{lemma}

\begin{lemma}

$(\theta_0, \beta_0) = \argmax M(\theta,\beta)$.
\end{lemma}

\begin{proof}
\begin{align*}
  &\mathbb{E}_{(\beta_0,\theta_0)} \Bigg[ ln
  \frac{f(A|\bm{X};\beta,\theta)}{f(A|\bm{X};\beta_0,\theta_0)}\Bigg]
  \leq ln \Bigg[\mathbb{E}_{(\beta_0,\theta_0)} \frac{f(A|\bm{X};\beta,\theta)}{f(A|\bm{X};\beta_0,\theta_0)}\Bigg] = 0\\
  \Longrightarrow & M(\theta,\beta) \leq M(\theta_0,\beta_0)
\end{align*}
\end{proof}

\begin{lemma}[Consistent treatment distribution estimation]\label{consistentA}
Assume $(\beta,\theta)$ is known, $\hat{Q}_{\theta}(\bm{x},a)$ and $\Tilde{f}_{A\mid \bm{X}}(a|\bm{x})$ are the estimated outcome and conditional treatment distribution, $Q_{\theta}(\bm{x},a)$ and $f_{A\mid \bm{X}}(a|\bm{x})$ are the ground truth. If $\sup_{a \in \mathcal{A}} \mid Q_{\theta}(\bm{x},a) - \hat{Q}_{\theta}(\bm{x},a)\mid = o_p(1)$, then $\forall$ $g(A,X)$ with $\mathbb{E}_{A|X,Q} \ |g(A,X)| < M < \infty$, we have 
$\mathbb{\hat{E}}_{A|X,\hat{Q}}\  g(A,X) \overset{p}{\to} \mathbb{E}_{A|X,Q} \ g(A,X)$.
\end{lemma}

\begin{proof}
If $\sup_{a \in \mathcal{A}} \mid Q_{\theta}(\bm{x},a) - \hat{Q}_{\theta}(\bm{x},a)\mid = \epsilon \to 0$, then 
$exp\{ \beta  [\hat{Q}_{\theta}(\bm{x},a) - Q_{\theta}(\bm{x},a)] \} \in [ exp(-\beta \epsilon), exp(\beta \epsilon)] \subset [1-\eta,1+\eta]$ $(\eta \to 0)$.
\begin{align*}
  &|\Tilde{f}_{A\mid \bm{X}}(A=a \mid \bm{X} = \bm{x}) -  f_{A\mid \bm{X}}(A=a \mid \bm{X} = \bm{x})| \\
  =& \mid
  \frac{exp\{ \beta  \hat{Q}_{\theta}(\bm{x},a) \}}{ \int_{\mathcal{A}} exp\{\beta   \hat{Q}_{\theta}(\bm{x},t) \} d\nu(t) } - \frac{exp\{ \beta  Q_{\theta}(\bm{x},a) \}}{ \int_{\mathcal{A}} exp\{\beta   Q_{\theta}(\bm{x},t) \} d\nu(t) }\mid\\
  =& \mid\frac{exp\{ \beta  Q_{\theta}(\bm{x},a) \} \  exp\{ \beta  [\hat{Q}_{\theta}(\bm{x},a) - Q_{\theta}(\bm{x},a)] \} }{ \int_{\mathcal{A}} exp\{\beta   Q_{\theta}(\bm{x},t) \} \ exp\{ \beta  [\hat{Q}_{\theta}(\bm{x},t) - Q_{\theta}(\bm{x},t)] \}d\nu(t) } - 
  \frac{exp\{ \beta  Q_{\theta}(\bm{x},a) \}}{ \int_{\mathcal{A}} exp\{\beta   Q_{\theta}(\bm{x},t) \} d\nu(t) }\mid\\
  \leq &  \ f_{A\mid \bm{X}}(A=a \mid \bm{X} = \bm{x})  \times max\{ \frac{2\eta}{1-\eta}, \frac{2\eta}{1+\eta} \}\\
  \leq & \ 4 \eta f_{A\mid \bm{X}}(A=a \mid \bm{X} = \bm{x})
\end{align*}

Thus if $\mathbb{E}_{A|X,Q} |g(A,X)| < M < \infty$, we have $(\mathbb{E}_{A|X,Q} - \mathbb{\hat{E}}_{A|X,\hat{Q}}) g(A,X) \le 4\eta M \overset{p}{\to} 0$.
    
\end{proof}

\begin{lemma}
Given $\theta$, $\mathbb{ E} \mid \hat{d}_{\theta}^{opt}(\bm{X}) - d_{\theta}^{opt}(\bm{X})  \mid  \xrightarrow{p} 0 $, where $\hat{d}_{\theta}^{opt}(\bm{x}) = \argmax_{a\in \mathcal{A}} \hat{Q}_{\theta}(\bm{x},a)$, 
$d_{\theta}^{opt}(\bm{x}) = \argmax_{a\in \mathcal{A}} Q_{\theta}(\bm{x},a)$.
\end{lemma}

\begin{proof}
    Because (i) $Q_{\theta}$ is identifiable (Assumption \ref{identifyQ}), (ii) $\hat{d}_{\theta}^{opt}(\bm{x}) = \argmax_{a\in \mathcal{A}} \hat{Q}_{\theta}(\bm{x},a)$ and (iii) $\sup_{a \in \mathcal{A}} \mid Q_{\theta}(\bm{x},a) - \hat{Q}_{\theta}(\bm{x},a)\mid \xrightarrow{p} 0$, uniformly for $\bm{x} \in \mathcal{X}$ (Assumption \ref{consistentQ}), we can apply lemma \ref{Kosorok2.12} to get the conclusion directly.
\end{proof}

\begin{proof}[\textbf{Proof of Theorem \ref{consistency}}]
    We firstly introduce the following notations: true score function is $m(\theta, \beta) := m(\bm{X},A; \theta, \beta) = \beta  Q_{\theta}(\bm{X},A) - 
    log\int_{\mathcal{A}} exp\{\beta   Q_{\theta}(\bm{X},t) \} d\nu(t)$, estimated score function with fitted $Q$ is $\hat{m}(\theta, \beta) := \hat{m}(\bm{X},A; \theta, \beta) = \beta  \hat{Q}_{\theta}(\bm{X},A) - 
    log\int_{\mathcal{A}} exp\{\beta   \hat{Q}_{\theta}(\bm{X},t) \} d\nu(t)$. The empirical expectation and maximal likelihood estimator are defined by $M_n(\theta, \beta) = \mathbb{E}_n \hat{m}(\theta, \beta) $ and $M_n( \hat{\theta}, \hat{\beta}) = \max M_n(\theta, \beta) $. The true expectation is $M( \theta, \beta) = \mathbb{E}  m(\theta, \beta) $. To complete the proof, we will take three steps.

    \paragraph{1,} $\sup_{(\theta,\beta)} \mid (\mathbb{E}_n - \mathbb{E}) \hat{m}(\bm{X},A; \theta, \beta) \mid \xrightarrow{p} 0$.

    Assumption \ref{continuous_log} is directly applied when treatment doses can take on any continuous value. In the case where treatment only consists of a finite set of dose options $(a_1,a_2,\dots,a_m)$, we can express $\hat{m}(\theta,\beta)$ as $ \beta  \hat{Q}_{\theta}(\bm{X},A) - 
    log \sum_{j=1}^{m} exp\{\beta   \hat{Q}_{\theta}(\bm{X},a_j) \} + log(m)$, and we will demonstrate that Assumption \ref{continuous_log} is automatically satisfied in this case.
    
    Because $\hat{Q}_{\theta}(\bm{X},A) = w(\bm{X};\theta)[\hat{Q}_Y(\bm{X},A) - \hat{Q}_Z(\bm{X},A)] + \hat{Q}_Z(\bm{X},A)$ lies in a VC class indexed by $\theta \in \mathbb{R}^p$ by Lemma 9.6 and Lemma 9.9 (viii), (vi) and (v) of Kosorok (2008). Similarly, $exp\{\beta \hat{Q}_{\theta}(\bm{X},a_j) \}, j \in \{1,2,...,m \}$, and their summation lie in a VC class indexed by $\theta \in \mathbb{R}^p$ as long as $\beta$ is bounded, which holds as long as $\mathcal{B}$ is compact. Then it follows that $\hat{m}(\bm{X},A; \theta, \beta)$ lies in a GC class indexed by $(\theta,\beta) \in \mathbb{R}^p \times \mathcal{B}$ and 
   $$\sup_{(\theta,\beta)\in \mathbb{R}^p \times \mathcal{B}} \Big| (\mathbb{E}_n - \mathbb{E}) \Big[ \beta  \hat{Q}_{\theta}(\bm{X},A) - 
    log\int_{\mathcal{A}} exp\{\beta   \hat{Q}_{\theta}(\bm{X},t) \} d\nu(t) \Big]  \Big| \xrightarrow{p} 0.$$

    \paragraph{2,} $\mathbb{E}  \{ m(\theta, \beta) -  \hat{m}(\theta, \beta) \} \xrightarrow{p} 0$.
    
    Denote $\Delta_{\theta}(\bm{X},A) = \hat{Q}_{\theta}(\bm{X},A) - Q_{\theta}(\bm{X},A)$. According to Assumption \ref{consistentQ}, we have that $\sup_{a} \mid \Delta_{\theta}(\bm{X},a)\mid = o_p(1)$:
    \vspace{-0.4cm}  
    \begin{align*}
        &\mid \hat{m}(\theta, \beta) - m(\theta, \beta) \mid \\
        =& | \beta \Delta_{\theta}(\bm{X},A) - log\int_{\mathcal{A}} exp\{\beta   Q_{\theta}(\bm{X},t)\}  exp\{\beta   \Delta_{\theta}(\bm{X},t)  \} d\nu(t)  + log\int_{\mathcal{A}} exp\{\beta   Q_{\theta}(\bm{X},t) \} d\nu(t) \mid\\
        \leq& \mid \beta \Delta_{\theta}(\bm{X},A) \mid  +  \mid log\  exp\{\sup_{a} \beta \Delta_{\theta}(\bm{X},a)\} |\\
        \leq& 2\beta \sup_{a} \mid \Delta_{\theta}(\bm{X},a)\mid \xrightarrow{p} 0.
    \end{align*}
Then 
\vspace{-0.8cm}  
\begin{align*}
        &\mid M(\theta,\beta) - M_n(\theta,\beta) \mid  
        =   \mid \mathbb{E}  m(\theta, \beta) - \mathbb{E}_n\hat{m}(\theta, \beta) \mid  \\
        =&   \mid \mathbb{E}  \{ m(\theta, \beta) -  \hat{m}(\theta, \beta) \} + (\mathbb{E} -\mathbb{E}_n ) \hat{m}(\theta, \beta) \mid  = o_p(1).
    \end{align*}

\paragraph{3,} $(\hat{\beta}_n,\hat{\theta}_n) \xrightarrow{p} (\beta_0,\theta_0)$.

Because $M_n( \hat{\theta}, \hat{\beta}) = \max M_n(\theta, \beta)$, $M( \theta_0, \beta_0) = \max M(\theta, \beta)$, $| M(\theta,\beta) - M_n(\theta,\beta) | = o_p(1)$, we can apply lemma \ref{Kosorok2.12} and Assumption \ref{identify} to prove the consistency of $(\hat{\beta}_n,\hat{\theta}_n)$.

\end{proof}

\begin{proof}[\textbf{Proof of Theorem \ref{value}}]

Denote $\hat{d}_{\hat{\theta}_n}(\bm{x}) = \argmax_{a} \hat{Q}_{\hat{\theta}_n}(\bm{x},a)$ and  $d_0(\bm{x}) = \argmax_{a} Q_{\theta_0}(\bm{x},a)$ as estimated optimal treatment policy and true optimal policy respectively. 

Consider the optimality of proposed treatment assignment $\hat{d}_{\hat{\theta}_n}(\bm{x})$ and compare it with the true optimal one.

\begin{align*}
    &\left| Q_{\theta_0}(\bm{x},d_0(\bm{x})) - Q_{\theta_0}(\bm{x},\hat{d}_{\hat{\theta}_n}(\bm{x})) \right|\\
    \leq\ & \left| Q_{\theta_0}(\bm{x},d_0(\bm{x})) - \hat{Q}_{\hat{\theta}_n}(\bm{x},\hat{d}_{\hat{\theta}_n}(\bm{x})) \right| 
    + 
    \left| \hat{Q}_{\hat{\theta}_n}(\bm{x},\hat{d}_{\hat{\theta}_n}(\bm{x})) - Q_{\theta_0}(\bm{x},\hat{d}_{\hat{\theta}_n}(\bm{x})) \right|\\
    \leq\ & 2\sup_a  \left| Q_{\theta_0}(\bm{x},a) - \hat{Q}_{\hat{\theta}_n}(\bm{x},a)  \right|\\
    \leq\ & 2\sup_a \Big[ \left| Q_{\theta_0}(\bm{x},a) - Q_{\hat{\theta}_n}(\bm{x},a) \right|
    +
     \left| Q_{\hat{\theta}_n}(\bm{x},a) - \hat{Q}_{\hat{\theta}_n}(\bm{x},a) \right| \Big] \\
    =\ &  2\sup_a \left|[w(\bm{x}; \theta_0 ) - w(\bm{x}; \hat{\theta}_n )][Q_Y(\bm{x},a) - Q_Z(\bm{x},a)]\right| + o_p(1)  \\
    =\ & o_p(1).
\end{align*}
Where the first and the third inequalities are using triangle inequality; the second inequality is because $|\sup_a f(a) - \sup_a g(a)| \leq \sup_a|f(a) - g(a)|$; the fourth equation results from consistency of the outcome model (Assumption \ref{consistentQ}); the last equation is because of the bounded first derivative of $w(\bm{x}; \theta)$ and consistency of $\hat{\theta}_n$. This property holds uniformly for $\forall \bm{x} \in \mathcal{X}$.

$$ V_{\theta_0}(d_0) -  V_{\theta_0}(\hat{d}_{\hat{\theta}_n}) = \mathbb{E}_{\bm{X}} \left[ Q_{\theta_0}(\bm{X},d_0(\bm{X})) - Q_{\theta_0}(\bm{X},\hat{d}_{\hat{\theta}_n}(\bm{X})) \right] = o_p(1)$$

\end{proof}

\begin{proof}[\textbf{Proof of Theorem \ref{normal}}]

We will split the log likelihood into two parts:
\vspace{-0.8cm}  
\begin{align}
  &\hat{l}_n(\hat{\theta}_n, \hat{\beta}_n) -   \hat{l}_n(\theta_0, \beta_0)  \\
  =\ & \sum_{i=1}^{n} \left[   \hat{\beta}_n \hat{Q}_{\hat{\theta}_n}(\bm{X}_i,A_i) -log \int_{\mathcal{A}} exp\{  \hat{\beta}_n \hat{Q}_{\hat{\theta}_n}(\bm{X}_i,z) \} dz 
  \right] - \hat{l}_n(\theta_0, \beta_0) \\
  =\ & \sum_{i=1}^{n}  \left[  \hat{\beta}_n \hat{Q}_{\hat{\theta}_n}(\bm{X}_i,A_i) - \beta_0 \hat{Q}_{\theta_0}(\bm{X}_i,A_i) \right]\\
  & -\ \sum_{i=1}^{n} \left[   log \int_{\mathcal{A}} exp\{  \hat{\beta}_n \hat{Q}_{\hat{\theta}_n}(\bm{X}_i,z) \} dz -
  log \int_{\mathcal{A}} exp\{ \beta_0 \hat{Q}_{\theta_0}(\bm{X}_i,z) \} dz 
  \right]
\end{align}

For the first part, we have that
\vspace{-0.4cm}  
\begin{align*}
&\sum_{i=1}^{n}    \hat{\beta}_n \hat{Q}_{\hat{\theta}_n}(\bm{X}_i,A_i) - \beta_0 \hat{Q}_{\theta_0}(\bm{X}_i,A_i)   \\
=\ &
\Big( 
\sum_{i=1}^{n}  \hat{Q}_{\theta_0}(\bm{X}_i,A_i) ,  \sum_{i=1}^{n} \beta_0 \hat{R}(\bm{X}_i,A_i)w'(\bm{X}_i;\theta_0)    
\Big) (\hat{\beta}_n - \beta_0, \hat{\theta}_n - \theta_0)^T\\
& + \frac{1}{2}(\hat{\beta}_n - \beta_0, \hat{\theta}_n - \theta_0)\begin{pmatrix}
0 & \sum_{i=1}^{n} \hat{R}_i w'_i(\theta) \\
\sum_{i=1}^{n}  \hat{R}_i w'_i(\theta) & 
\sum_{i=1}^{n} \beta \hat{R}_i w''_i(\theta)
\end{pmatrix}\Bigg|_{ \beta= \beta_1^*,\theta= \theta_1^*}
\begin{pmatrix}
\hat{\beta}_n - \beta_0 \\
 \hat{\theta}_n - \theta_0
\end{pmatrix}
\end{align*}
where $(\beta_1^*,\theta_1^*)$ are intermediate values between $(\beta_0,\theta_0)$ and $(\hat{\beta}_n,\hat{\theta}_n)$ according to two variable Taylor expansion, $\hat{R}_i = \hat{R}(\bm{X}_i,A_i)$, $w'_i(\theta^*) = w'(\bm{X}_i;\theta^*)$,$w''_i(\theta^*) = w''(\bm{X}_i;\theta^*)$,

Similarly, for the second part:
\vspace{-0.4cm}  
\begin{align*}
  &  \sum_{i=1}^{n} \left[   log \int_{\mathcal{A}} exp\{  \hat{\beta}_n \hat{Q}_{\hat{\theta}_n}(\bm{X}_i,z) \} dz -
  log \int_{\mathcal{A}} exp\{ \beta_0 \hat{Q}_{\theta_0}(\bm{X}_i,z) \} dz 
  \right] 
  = F(\hat{\beta}_n,\hat{\theta}_n) - F(\beta_0,\theta_0) \\
  =\ & (\frac{\partial F}{\partial \beta}, \frac{\partial F}{\partial \theta} )\bigg|_{\beta = \beta_0, \theta= \theta_0} (\hat{\beta}_n - \beta_0, \hat{\theta}_n - \theta_0)^T + \frac{1}{2}(\hat{\beta}_n - \beta_0, \hat{\theta}_n - \theta_0)
\begin{pmatrix}
F_{\beta \beta} & F_{\beta \theta} \\
 F_{\beta \theta} & F_{\theta \theta}
\end{pmatrix}\Bigg|_{ \beta= \beta_2^*,\theta= \theta_2^*}
\begin{pmatrix}
\hat{\beta}_n - \beta_0 \\
 \hat{\theta}_n - \theta_0
\end{pmatrix}
\end{align*}
where $(\beta_2^*,\theta_2^*)$ are intermediate values between $(\beta_0,\theta_0)$ and $(\hat{\beta}_n,\hat{\theta}_n)$ according to two variable Taylor expansion. The derivative functions are as follows.

\vspace{-0.4cm}  

$$\frac{\partial F}{\partial \beta} = 
 \sum_{i=1}^{n}  \mathbb{E}_{A\mid \bm{X}_i,\beta,\theta, \hat{Q}} \left[ \hat{Q}_{\theta}(\bm{X}_i,A) \mid \bm{X}_i \right],\ \frac{\partial F}{\partial \theta}  = 
 \sum_{i=1}^{n}  \mathbb{E}_{A\mid \bm{X}_i,\beta,\theta, \hat{Q}} \left[ \beta \hat{R}(\bm{X}_i,A)w'(\bm{X}_i;\theta)   \mid \bm{X}_i \right]$$
\vspace{-0.4cm}  
$$F_{\beta\beta} = \frac{\partial^2 F}{\partial \beta^2}   =  \sum_{i=1}^{n}  var_{A\mid \bm{X}_i,\beta,\theta, \hat{Q}} \left[ \hat{Q}_{\theta}(\bm{X}_i,A)   \mid \bm{X}_i \right]$$
\vspace{-0.4cm}  
$$F_{\theta\theta} =\frac{\partial^2 F}{\partial \theta^2}  = 
 \sum_{i=1}^{n}  var_{A\mid \bm{X}_i,\beta,\theta, \hat{Q}} \left[ \beta \hat{R}_iw'(\theta)   \mid \bm{X}_i \right] + 
 \sum_{i=1}^{n}  \mathbb{E}_{A\mid \bm{X}_i,\beta,\theta, \hat{Q}} \left[ \beta \hat{R}_iw''(\theta)   \mid \bm{X}_i \right]$$
\vspace{-0.4cm}  
 $$F_{\beta\theta} = \frac{\partial^2 F}{\partial \theta\partial \beta}  = 
 \sum_{i=1}^{n}  cov_{A\mid \bm{X}_i,\beta,\theta, \hat{Q}} \left[ \hat{Q}_{\theta}(\bm{X}_i,A), \beta \hat{R}_iw'(\theta)   \mid \bm{X}_i \right] + 
 \sum_{i=1}^{n}  \mathbb{E}_{A\mid \bm{X}_i,\beta,\theta, \hat{Q}} \left[ \hat{R}_iw'(\theta)   \mid \bm{X}_i \right]$$

Get back to the original log likelihood expression, we rephrase it as a quadratic function of $(\hat{\beta}_n,\hat{\theta}_n)$.
\vspace{-0.4cm}  
\begin{align*}
  &\hat{l}_n(\hat{\theta}_n, \hat{\beta}_n) -   \hat{l}_n(\theta_0, \beta_0)  \\
=\ & \Big( 
\sum_{i=1}^{n}  \hat{Q}_{\theta_0}(\bm{X}_i,A_i) ,  \sum_{i=1}^{n} \beta_0 \hat{R}(\bm{X}_i,A_i)w'(\bm{X}_i;\theta_0)    
\Big) (\hat{\beta}_n - \beta_0, \hat{\theta}_n - \theta_0)^T\\
&+ \frac{1}{2}(\hat{\beta}_n - \beta_0, \hat{\theta}_n - \theta_0)
\begin{pmatrix}
0 & \sum_{i=1}^{n} \hat{R}_i w'_i(\theta) \\
\sum_{i=1}^{n}  \hat{R}_i w'_i(\theta) & 
\sum_{i=1}^{n} \beta \hat{R}_i w''_i(\theta)
\end{pmatrix}\Bigg|_{ \beta= \beta_1^*,\theta= \theta_1^*}
\begin{pmatrix}
\hat{\beta}_n - \beta_0 \\
 \hat{\theta}_n - \theta_0
\end{pmatrix}\\
&- (\frac{\partial F}{\partial \beta}, \frac{\partial F}{\partial \theta} )\bigg|_{ \theta= \theta_0} (\hat{\beta}_n - \beta_0, \hat{\theta}_n - \theta_0)^T \\
&- \frac{1}{2}(\hat{\beta}_n - \beta_0, \hat{\theta}_n - \theta_0)
\begin{pmatrix}
F_{\beta \beta} & F_{\beta \theta} \\
 F_{\beta \theta} & F_{\theta \theta}
\end{pmatrix}\Bigg|_{ \beta= \beta_2^*,\theta= \theta_2^*}
\begin{pmatrix}
\hat{\beta}_n - \beta_0 \\
 \hat{\theta}_n - \theta_0
\end{pmatrix}\\
=\  & (A_1,A_2) (\sqrt{n}(\hat{\beta}_n - \beta_0), \sqrt{n}(\hat{\theta}_n - \theta_0))^T  \\
& - \frac{1}{2}(\sqrt{n}(\hat{\beta}_n - \beta_0), \sqrt{n}(\hat{\theta}_n - \theta_0)) \begin{pmatrix}
B_{11} & B_{12} \\
B_{21} & B_{22}
\end{pmatrix}(\sqrt{n}(\hat{\beta}_n - \beta_0), \sqrt{n}(\hat{\theta}_n - \theta_0))^T
\end{align*}

$$A_1 =  n^{-\frac{1}{2}} \sum_{i=1}^{n}  \Big\{ \hat{Q}_{\theta_0}(\bm{X}_i,A_i)  -     \mathbb{E}_{A\mid \bm{X}_i,\beta_0,\theta_0, \hat{Q}} \left[ \hat{Q}_{\theta_0}(\bm{X}_i,A) \mid \bm{X}_i \right]  \Big\} $$
$$A_2 = n^{-\frac{1}{2}} \sum_{i=1}^{n} \Big\{ \beta_0 \hat{R}(\bm{X}_i,A_i)w'(\bm{X}_i;\theta_0) - \mathbb{E}_{A\mid \bm{X}_i,\beta_0,\theta_0, \hat{Q}} \left[ \beta_0 \hat{R}(\bm{X}_i,A)w'(\bm{X}_i;\theta)   \mid \bm{X}_i \right] \Big\} $$
$$B_{11} = \frac{1}{n}   \sum_{i=1}^{n}  var_{A\mid \bm{X}_i,\beta_2^*,\theta_2^*, \hat{Q}} \left[ \hat{Q}_{\theta}(\bm{X}_i,A)   \mid \bm{X}_i \right]$$
$$B_{22} = \frac{1}{n}\sum_{i=1}^{n}  var_{A\mid \bm{X}_i,\beta_2^*,\theta_2^*, \hat{Q}} \left[ \beta \hat{R}_iw'(\theta)   \mid \bm{X}_i \right] - \frac{1}{n} \sum_{i=1}^{n}  
\Big\{ \beta_1^* \hat{R}_iw''(\theta_1^*) -  \mathbb{E}_{A\mid \bm{X}_i,\beta_2^*,\theta_2^*, \hat{Q}} \left[ \beta \hat{R}_iw''(\theta)   \mid \bm{X}_i \right] \Big\} $$
$$B_{12} =  \frac{1}{n} \sum_{i=1}^{n}  cov_{A\mid \bm{X}_i,\beta_2^*,\theta_2^*, \hat{Q}} \left[ \hat{Q}_{\theta}(\bm{X}_i,A), \beta \hat{R}_iw'(\theta)   \mid \bm{X}_i \right] - \frac{1}{n} \sum_{i=1}^{n}  \Big\{ \hat{R}_iw'(\theta_1^*) - \mathbb{E}_{A\mid \bm{X}_i,\beta_2^*,\theta_2^*, \hat{Q}} \left[ \hat{R}_iw'(\theta)   \mid \bm{X}_i \right] \Big\}$$

Because of the convergence of $(\hat{\beta}_n,\hat{\theta}_n)$, both  $(\beta_1^*,\theta_1^*)$ and $(\beta_2^*,\theta_2^*)$ converge to $(\beta_0,\theta_0)$. Along with the twice continuously differentiable assumption of $w''(\theta)$, law of large number and lemma \ref{consistentA},

$$\begin{pmatrix}
B_{11} & B_{12} \\
B_{21} & B_{22}
\end{pmatrix} \overset{p}{\to}  B = \mathbb{E}_{\bm{X}}  Cov_{A} \Bigg[
    \begin{array}{c}
        Q_{\theta_0}(\bm{X},A) \\
 \beta_0R(\bm{X},A)w'(\bm{X};\theta_0)
    \end{array}\bigg|\  \bm{X} \Bigg]$$

Based on the quadratic expression of after Taylor expansion, the estimation $(\hat{\beta},\hat{\theta})$ is asymptotic linear: 
$$\sqrt{n}\ \Bigg( 
    \begin{array}{lc}
        \hat{\beta}_n - \beta_0 \\
 \hat{\theta}_n - \theta_0
    \end{array}
\Bigg) =  B^{-1}
\Bigg( 
    \begin{array}{c}
       n^{-\frac{1}{2}}  \sum_{i=1}^n \left\{ Q_{\theta_0}(\bm{X}_i,A_i) -\mathbb{E}\left[  Q_{\theta_0}(\bm{X}_i,A)   \mid \bm{X}_i \right] \right\} \\
n^{-\frac{1}{2}} \sum_{i=1}^n \beta_0 \left\{ R(\bm{X}_i,A_i)w'(\bm{X}_i;\theta_0) -\mathbb{E}\left[R(\bm{X}_i,A_i)w'(\bm{X}_i;\theta_0)   \mid \bm{X}_i \right] \right\}
    \end{array}
\Bigg) + o_p(1)
$$

And the asymptotic distribution of $(\hat{\beta}_n ,
 \hat{\theta}_n )$ is
$\sqrt{n}( \hat{\beta}_n - \beta_0, \hat{\theta}_n - \theta_0)
\sim N(0, B^{-1}) $.

\end{proof}

\bibliographystyle{apalike}
\bibliography{references}

\end{document}